\documentclass[draftcls,onecolumn]{IEEEtran}

\usepackage[utf8]{inputenc} 
\usepackage[T1]{fontenc}    
\usepackage{hyperref}       
\usepackage{url}            
\usepackage{booktabs}       
\usepackage{amsmath,amssymb,amsfonts,amsthm}       
\usepackage{nicefrac}       
\usepackage{microtype}      
\usepackage{xcolor}         

\usepackage{dsfont}
\usepackage{algorithm}
\usepackage{algorithmic}
\usepackage{graphicx}
\usepackage{textcomp}
\usepackage{subcaption}
\usepackage{pgfplots}
\pgfplotsset{compat=1.12}
\usepackage{tikz}
\usepackage{verbatim}
\usetikzlibrary{shapes,arrows}

\usetikzlibrary{positioning}
\usetikzlibrary{calc}

\usepackage{graphics,hyperref}
\usepackage{wrapfig, framed, caption}
\usepackage{color}
\usepackage{bbm}
\usepackage{array}
\newcolumntype{P}[1]{>{\centering\arraybackslash}p{#1}}
\usepackage{multirow}
\usepackage{multicol}
\usepackage{enumitem}

\newcommand{\vs}{\vspace{-0.1in}}

\newcommand{\calD}{\mathcal{D}}

\newcommand{\bfB}{\mathbf{B}}

\newcommand{\bfI}{\mathbf{I}}

\newcommand{\bfR}{\mathbf{R}}

\newtheorem{claim}{Claim}

\newtheorem{remark}{Remark}

\usepackage{mathtools}

{\begin{list}{}%
         {\setlength{\leftmargin}{#1}}%
         \item[]%
}
{\end{list}}

\title{Leveraging partial stragglers within gradient coding}

\author{%
  Aditya Ramamoorthy \quad Ruoyu Meng \quad Vrinda S. Girimaji\\
  Department of Electrical and Computer Engineering\\
  Iowa State University\\
  Ames, IA 50010. \\
  \texttt{\{adityar, rmeng, vrindasg\}@iastate.edu} 
}

\begin{document}

\maketitle

\begin{abstract}
Within distributed learning, workers typically compute gradients on their assigned dataset chunks and send them to the parameter server (PS), which aggregates them to compute either an exact or approximate version of $\nabla L$ (gradient of the loss function $L$). However, in large-scale clusters, many workers are slower than their promised speed or even failure-prone. A gradient coding solution introduces redundancy within the assignment of chunks to the workers and uses coding theoretic ideas to allow the PS to recover $\nabla L$ (exactly or approximately), even in the presence of stragglers. Unfortunately, most existing gradient coding protocols are inefficient from a computation perspective as they coarsely classify workers as operational or failed; the potentially valuable work performed by slow workers (partial stragglers) is ignored. In this work, we present novel gradient coding protocols that judiciously leverage the work performed by partial stragglers. Our protocols are efficient from a computation and communication perspective and numerically stable. For an important class of chunk assignments, we present efficient algorithms for optimizing the relative ordering of chunks within the workers; this ordering affects the overall execution time. For exact gradient reconstruction, our protocol is around $2\times$ faster than the original class of protocols and for approximate gradient reconstruction, the mean-squared-error of our reconstructed gradient is several orders of magnitude better.

\end{abstract}

\section{Introduction}
Large scale distributed learning is the workhorse of modern day machine learning (ML) algorithms. The sheer size of the data and the corresponding computation needs, necessitate the usage of huge clusters for the purpose of parameter fitting in most ML problems of practical interest: deep learning \cite{GoodBengCour16}, low-rank matrix completion \cite{jain2013lrmc} etc.  A typical scenario consists of a dataset $\mathcal{D} = \{(\mathbf{x}_i, y_i)\}_{i = 1}^{\tilde{N}}$ of $\tilde{N}$ data points, where $\mathbf{x}_i$'s and $y_i$'s are the features and labels respectively. We wish to minimize a loss function $ L = \frac{1}{\tilde{N}}\sum_{i = 1}^{\tilde{N}} l(\mathbf{x}_i, y_i, \mathbf{w})$ with respect to $\mathbf{w} \in \mathbb{R}^d$ ($\mathbf{w}$: parameter vector, $l$: prediction error). 
When $\calD$ is large, we can perform the learning task in a distributed manner \cite{li2014scaling}. 

{\noindent {\bf Background:}} We partition $\calD$ into $N$ equal-sized chunks denoted $\calD_i, i \in [N]$ ($[n]$ denotes the set $\{1, \dots, n\}$), where a chunk is a subset of the data points and distinct chunks are disjoint. Within each chunk, the assignment of the data points to the workers is identical. Suppose that there are $m$ workers $W_1, \dots, W_m$ and a parameter server (PS). We distribute the chunks to the different workers and let them compute the gradients on the data points assigned to them. The PS coordinates the training by aggregating the (partial) gradients from the workers and transmitting an updated 
parameter vector to the workers at each iteration. In the ``baseline'' scheme, $N = m$, $W_j$ is assigned $\calD_j$ and it computes $\sum_{i \in \mathcal{D}_j}  \nabla l(\mathbf{x}_i, y_i, \mathbf{w}_t)$ (a vector of length-$d$) and sends it to the PS. Using these, the PS computes the desired gradient $\nabla L = \frac{1}{\tilde{N}}\sum_{i=1}^{\tilde{N}} \nabla l(\mathbf{x}_i, y_i, \mathbf{w}_t)$ and the updated parameter  $\mathbf{w}_{t+1}$ thereafter. Unfortunately, in many large scale clusters, workers are often slower than their promised speed or even prone to failure. This is especially true in cloud platforms, where the load often fluctuates depending on system load and spot instance pricing \cite{GoogleCP} explicitly builds in the possibility of job preemption. To address these issues, gradient coding (GC) introduced in \cite{tandon_gradient} incorporates redundancy within the assignment of chunks to the workers. Once a given worker calculates the gradient on all its assigned chunks, it computes a specified linear combination of these gradients and sends it to the PS. An exact gradient coding solution allows the PS to exactly recover $\nabla L$ even in the presence of limited node failures. 

Let $A \in \mathbb{R}^{N \times m}$ be a matrix such that $A_{i,j} \neq 0$ if and only if $\calD_i$ is assigned to $W_j$; henceforth, we call this the assignment matrix. Let $nnz(v)$ denote the number of non-zero entries in a vector $v$. Let $\gamma_i$  and $\delta_j$ denote $nnz(A(i,:)$ and $nnz(A(:,j))$ (using MATLAB notation). These correspond respectively to the number of times $\calD_i$ appears in the cluster (replication factor) and the number of chunks assigned to $W_j$ (load factor). We assume that at most $s$ workers out of $m$ are stragglers. Let $g_{\calD_j} = \sum_{i \in \calD_j} \nabla l(\mathbf{x}_i, y_i, \mathbf{w}_t)$, so that $\nabla L = \sum_{j=1}^N g_{\calD_j}$. At iteration $t$, worker $W_j$ calculates $g_{\calD_i}$ for all $i \in \text{supp}(A(:,j))$ (non-zero entries in $A(:,j)$) and linearly combines them to obtain $g_j = \sum_{i =1}^N A_{i,j} g_{\calD_i}.$
It subsequently transmits $g_j$ to the PS. For decoding $\nabla L$, the PS picks a decoding vector $r$ which is such that $r_j=0$ if $W_j$ has not transmitted $g_j$ (we say that $W_j$ is straggling in this case). It subsequently calculates
\begin{align}
\sum_{j=1}^m r_j g_j &= \sum_{i=1}^N  \left( \sum_{j=1}^m   A_{i,j} r_j \right) g_{\calD_i}. \label{eq:grad_cod_eq}
\end{align}
\noindent {\bf Related Work:} Under the original GC model \cite{tandon_gradient}, for {\it exact gradient coding}, we want that $Ar = \mathds{1}$ (the all-ones vector) under any choice of at most $s$ stragglers. This means that the PS can perform a full gradient update. For exact GC, we need $\gamma_i \geq s+1$ for all $i \in [N]$. This setting has been studied extensively \cite{dimakis_cyclic_mds,reedsolomon_GC18,chen2018lag,tieredGC20,dynamicClusteringGC23,numerically_stableGC20}. {\it Approximate gradient coding} \cite{dimakis_cyclic_mds} considers the scenario where the full gradient is either not required (e.g., SGD \cite{bottou_optimization} works with an inexact gradient) or too costly to work with because of the high replication factor needed. In this setting, we want to design $A$ such that $|| A r- \mathds{1}||_2$ ($\ell_2$-norm) is small over all possibilities for the straggling workers. Prior work demonstrates constructions from expander graphs \cite{dimakis_cyclic_mds}, sparse random graphs \cite{charles2017approximate} and \cite{glasgowW21}, block designs \cite{kadheKR19, soft_BIBD_GC22} and the like. Within distributed training, a significant time cost is associated with the transmission of the computed gradients (vectors of length-$d$) by the workers to the PS \cite{alistarh2017qsgd,wang2023cocktailsgd}, e.g., deep learning usually operates in the highly over-parameterized regime \cite{Belkin_2019} ($d \gg \tilde{N}$).  For exact GC, if $\gamma_i \geq s + \ell$ for $i \in [N]$, then the dimension of the transmitted vectors from the workers can be lowered to $d/ \ell$ \cite{YeA18,tayyebehM21}, thus saving on communication time. This is referred to as {\it communication-efficient GC} and allows us to trade-off communication for computation. However, both \cite{YeA18} and \cite{tayyebehM21} use polynomial interpolation in their solution. This leads to significant numerical instability \cite{Pan16} to the extent that their solution is essentially unusable for systems with twenty or more workers. This point is also acknowledged within the papers: Section V of \cite{tayyebehM21} and Section II of \cite{YeA18}. Some work considering these issues appears in \cite{kadheKR20} under restrictive parameter assumptions. 

The usage of the partial work performed by stragglers has been considered \cite{dynamicClusteringGC23,TauzD19,adaptiveGC22, hetero_aware_GC22,opt_block_coord_GC21,anytime_GC17} only within exact GC (i.e., approximate GC has not been considered); some of these apply in the communication-efficient setting. However, these approaches use multiple messages from the workers to the PS and thus incur a higher communication cost per iteration. 

\begin{figure}[t!]
    \hspace{-0.1in}
    \begin{subfigure}[t]{0.45\textwidth}
        \centering
        \includegraphics[scale = 0.7]{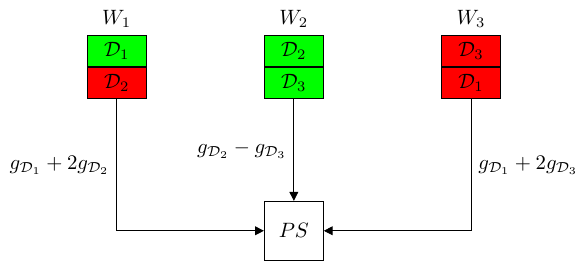}
        \caption{\label{fig:original_GD_2}}
    \end{subfigure}
    \hspace{0.1in}
    \begin{subfigure}[t]{0.45\textwidth}
        \centering
        \includegraphics[scale=0.7]{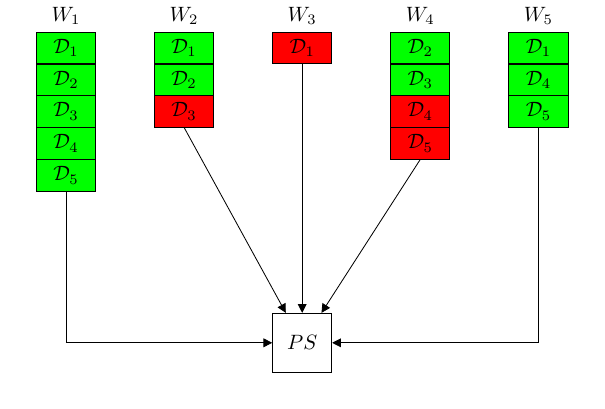}
        \caption{\label{fig:partial_straggler_GD}}
    \end{subfigure}

    \caption{{\small Green/red means that the worker did/didn't process a chunk. (a) System with $N=m=3$. Each worker is assigned two chunks that they process in a top-to-bottom order.  $W_3$ is failed and $W_1$ is slow. (b) An arbitrary assignment of chunks to the workers (example also appears in \cite{tayyebehM21}).}}
    %
\end{figure}



{\noindent {\bf Motivation:}} Consider Fig. \ref{fig:original_GD_2} where the edge labels indicate the encoded gradients. Note that under ideal operating conditions when each worker operates at the same speed, the overall gradient can be computed as long as each $W_i$ processes its first chunk $\calD_i$ for $i=1, \dots, 3$. Thus, it is quite wasteful for the workers to continue processing their second assigned chunk as specified in the original GC scheme; it would double the computation time. In addition, the original GC formulation ignores partial work by slow but not failed workers; in the sequel, we refer to these as ``partial stragglers''. For instance, in Fig. \ref{fig:original_GD_2}, we consider a scenario, where $W_1$ is slow and $W_3$ is failed. The state of the computation is such that there is enough information for the PS to obtain the full gradient. However, under the original GC model, $W_1$ will either wait until it processes $\calD_2$ before generating the encoded gradient to be sent to the PS, or the PS will treat $W_1$ as failed and will have to settle for an approximate gradient. 
Gradient coding can be viewed as an application of coding-theoretic techniques \cite{lincostello} to distributed learning; it allows the PS to be resilient to failures with very little coordination in the cluster. We note here that within classical coding theory \cite{lincostello} most constructions of erasure codes do not consider feedback from the receiver to the sender since such feedback may be expensive or noisy when the sender and receiver are at remote locations or not even necessary (\cite{CoverThomas2006}, Chap. 7). However, feedback is quite easy to implement in the distributed learning setup.


{\noindent {\bf Main Contributions:}} 
We present a new GC protocol that exploits a small amount of additional interactive communication to/from the PS and the workers. It  greatly improves the computation-efficiency and communication-efficiency of GC while continuing to allow for efficient coordination within the cluster. Specifically, our protocol efficiently leverages the chunks processed by partial stragglers. Prior work that potentially applies in our setting suffers from the serious problem of numerical instability. In contrast, our protocol is provably numerically stable.

To our best knowledge, despite the importance of communication-efficiency, there are hardly any schemes for communication-efficient approximate GC. Our protocol provides an elegant way to address this problem, which in addition allows the PS to obtain an accurate estimate of the mean-square-error of the reconstructed gradient at any given point in the computation.

Prior work in the GC area ignores the relative ordering of the chunks within the workers. As our protocol leverages partial work performed by the workers, the relative ordering of the chunks within the workers is an important factor in the performance of the algorithm. For a large class of assignment matrices, we present an efficient polynomial-time algorithm that returns the optimal ordering of the chunks within workers.



\section{Gradient Coding for partial stragglers}

\label{sec:gc_partial_stragglers}
Our GC protocol operates under the following assumptions: (i) the workers know the assignment matrix and the ordering of the chunks within all the workers, (ii) at regular intervals the workers keep communicating to the PS, the number of chunks they have processed, and (iii) the PS wants the workers to communicate vectors of length $d/\ell$ for integer $\ell \geq 1$.


We now provide a top-level description of our protocol; a formal statement appears in Algorithm \ref{alg:find-enc-coeff}.
At the beginning of the training process, the PS generates a random $\ell \times m$ matrix $\mathbf{R}$ with i.i.d. entries from a standard normal distribution, $N(0,1)$ and shares it with all the workers. The workers keep reporting the number of chunks that they have processed in an iteration. The PS keeps monitoring these counts, and at a certain point, it broadcasts to all the workers an ``encode-and-transmit'' signal and an integer vector $\psi$ of length-$m$  that specifies the number of chunks that have been processed by each node. This shares the global system state amongst the workers. For exact GC, the PS sends the encode-and-transmit signal when at least $\ell$ copies of each $\calD_i$ have been processed across the cluster. For approximate GC, the PS can send the signal even earlier.

Following this, the workers need to decide their own encoding coefficients for the gradients that they have computed. Each gradient $g_{\calD_i}$ is block-partitioned into $\ell$ disjoint parts $g_{\calD_i}[k], k = 0, \dots, \ell-1$. Each worker forms a matrix $\bfB$ of dimension $m \times N \ell$ that consists of indeterminates that specify the encoding coefficients of all the workers (see example in the upcoming Section \ref{sec:illustrative_eg}). Matrix $\bfB$ consists of $N$ block-columns, where each block-column itself consists of $\ell$ columns of length-$m$, i.e., $\bfB = [\bfB^{(1)}~|~\bfB^{(2)}~|~ \dots ~|~ \bfB^{(N)}]$. The $j$-th block column, $\bfB^{(j)}$ is associated with $g_{\calD_j}[k], k = 0, \dots, \ell-1$. Based on the global state vector $\psi$, all workers know whether a chunk $\calD_j, j\in [N]$ has been processed by worker $W_i, i\in [m]$. If $\calD_j$ has been processed by $W_i$, then the $i$-th row of $\bfB^{(j)}$ is populated with indeterminates, otherwise these entries are set to zero. Once the indeterminates are found (see Algorithm \ref{alg:find-enc-coeff} for a description), worker $W_{\beta}$ encodes its gradients as
\begin{align}
    g_\beta &= \sum_{i=1}^N \sum_{j=0}^{\ell-1} \bfB^{(i)}_{\beta, j} g_{\calD_i}[j].
\end{align}

Let $\vec{e}_i$ denote the $i$-th canonical basis vector of length-$\ell$, i.e., it contains a $1$ at location $i \in \{0, \dots, \ell-1\}$ and zeros elsewhere. We denote
\begin{align}
    z_i &=  \mathds{1}_{N} \otimes \vec{e}_i = \overbrace{[\vec{e}_i^T~\vec{e}_i^T~\dots~\vec{e}_i^T]^T}^{\text{$N$ copies of $\vec{e}_i$}}. \label{eq:kron_prod_z}
\end{align}
where $\otimes$ denotes the Kronecker product and $\mathds{1}_N$ denotes the all-ones vector of length $N$. Recall that in the exact GC scenario, the PS wants to obtain $\sum_{i=1}^N g_{\calD_i}[k]$ for $k = 0, \dots, \ell-1$. This is equivalent to requiring that
\begin{align*}
    z_i^T \in \text{row-span}\left( \bfB \right), \text{~for~} i = 0 , \dots, \ell-1.
\end{align*}
Our protocol is such that each $W_i$ can independently calculate their encoding coefficients, such that collectively all the workers agree on the same matrix $\bfB$ with the same values assigned to all the indeterminates.
Towards this end, let $\tilde{\bfB}_j$ denote the submatrix of $\bfB$ that is relevant to $W_j$, i.e., the block-columns in which the processed chunks of $W_j$ participate. For instance, suppose that $W_j$ has processed $\alpha_j \leq \delta_j$ chunks $\calD_{i_1}, \dots, \calD_{i_{\alpha_j}}$ where $1 \leq i_1 < i_2 < \dots < i_{\alpha_j} \leq N$. Then, 
\begin{align}
    \tilde{\bfB}_j &= [\bfB^{(i_1)}~\bfB^{(i_2)}~\dots~\bfB^{(i_{\alpha_j})} ]. \label{eq:extracted_block_cols}
\end{align}
$W_j$ then solves a minimum-$\ell_2$-norm least-squares solution to determine its encoding coefficients (see \eqref{eq:ls_min_algo} in Algorithm \ref{alg:find-enc-coeff}). This ensures the solution is unique \cite{horn_matrix_analysis} even if the corresponding problem is under-determined. Thus, the workers automatically agree on the same solution.


\begin{algorithm}[!t]
\caption{Find-Encoding-Coeff}
\begin{algorithmic}[1]
\REQUIRE $\ell \times m$ matrix $\bfR$, $m$-length state vector $\psi$, $\delta_i$ the number of chunks processed by $W_i$.
\ENSURE Encoding coefficients for the $i$-th worker $\varepsilon_i$.
\STATE $W_i$ forms the indeterminate matrix $m \times N \ell$ matrix $\bfB$ based on $\psi$.
\STATE $W_i$ extracts the columns of $\bfB$ that correspond to its processed chunks. This matrix is called $\tilde{\bfB}_i$ ({\it cf.} \eqref{eq:extracted_block_cols}).
\STATE Worker $W_i$ solves the following minimum-$\ell_2$-norm least-squares problem, where the variables are the indeterminates in $\tilde{\bfB}_i$.
\begin{align}
    \min ||   \mathds{1}_{\delta_i}^T \otimes I_\ell  - \bfR \tilde{\bfB}_i||_2. \label{eq:ls_min_algo}
\end{align}
\STATE Set $\varepsilon_i = \bfB(i,:)$.
\end{algorithmic}
\label{alg:find-enc-coeff}
\end{algorithm}
\begin{remark}
    If $\ell=1$, then it is possible to arrive at a protocol whereby the workers agree to transmit appropriately weighted partial sums of their gradients, so that the PS can exactly/approximately recover the sum. However, in the communication-efficient setting when $\ell > 1$, the encoding is no longer straightforward. Thus, $\ell > 1$ is the main scenario we consider in what follows.
\end{remark}
\begin{remark}
    We discussed that each worker can form the $m \times N\ell$ matrix of indeterminates $\bfB$ for the sake of ease of explanation. In fact, $W_j$ only works with $\tilde{\bfB}_j$, which is of size at most $m \times \delta \ell$.
\end{remark}
\begin{remark}
The additional communication assumed in our protocol is only $O(m)$ as against the parameter vector of length-$d$ and $O(m) \ll d$. Thus, the additional communication cost of our algorithm is minimal.
\end{remark}

\subsection{Analysis of Algorithm \ref{alg:find-enc-coeff}}
It is evident from our description and Algorithm \ref{alg:find-enc-coeff} that upon receiving the encode-and-transmit signal and the vector $\psi$, each worker can independently create the matrix of indeterminates $\bfB$.

{\bf Exact GC analysis:} Suppose that each $\calD_i, i \in [N]$ has been processed at least $\ell$ times across the cluster, and suppose that $W_i$ has processed $\calD_j$. This means that there is a block-column $\bfB^{(j)}$ such that each column within $\bfB^{(j)}$ has at least $\Delta \geq \ell$ indeterminates that need to be assigned values. For one column within $\bfB^{(j)}$, $W_i$ will solve a system of equations that is specified by $\ell \times \Delta$ submatrix of $\bfR$ denoted $X$ (an example appears in Section \ref{sec:illustrative_eg}). Note that the columns of $X$ will be linearly independent with high probability owing to the choice of $\bfR$ and since $\Delta \geq \ell$, a solution is guaranteed to exist. Let $\kappa_2(M)$ denote the condition number of matrix $M$. For a random $\ell \times \Delta$ matrix with i.i.d. $N(0,1)$ entries, it is known that $\mathbb{E} (\log \kappa_2(X)) \leq O(\log \Delta)$ \cite{ChenD05}. Thus, each such system of equations is well-conditioned with very high probability. In the under-determined case when $\Delta > \ell$, each worker will still agree on the same values of the corresponding indeterminates since we enforce that we work with the (unique) minimum $\ell_2$-norm solution; no additional communication between the workers is required. 

{\bf Approximate GC analysis:} It is possible that the PS sends the encode-and-transmit vector when some $\calD_i$ has been processed $\Delta \leq \ell-1$ times. In this case, the corresponding step for $\bfB^{(i)}$ in \eqref{eq:ls_min_algo} will be an over-determined least-squares procedure, which implies that there will be a non-zero error associated with it. Let $X$ be the relevant $\ell \times \Delta$ submatrix of $\bfR$ where we have $\Delta < \ell$ now, and recall that all entries of $X$ are i.i.d. $N(0,1)$ random variables. The squared error corresponding to $\bfB^{(i)}$ can be expressed as
$\sum_{i=0}^{\ell-1} ||XX^{\dagger} \vec{e}_i - \vec{e}_i||_2^2$ ($X^\dagger$ denotes the pseudo-inverse \cite{horn_matrix_analysis}). Let $X = U S V^T$ denote the SVD of $X$, where $U$ and $V$ are orthogonal matrices of dimension $\ell \times \ell$ and $\Delta \times \Delta$ respectively and $S = [D ~|~ 0]^T$ where $D$ is a $\Delta \times \Delta$ matrix with positive entries on the diagonal, and $0$ represents a $\Delta \times (\ell-\Delta)$ matrix of zeros. Then, $X^\dagger = V [D^{-1} ~|~ 0] U^T$. It is well known (\cite{versh_book}, Remark 5.2.8) that for a matrix with i.i.d. $N(0,1)$ entries, the singular vectors are uniformly distributed on the unit-sphere. Therefore, the expected squared error becomes
\begin{align}
    \mathbb{E}[||XX^{\dagger} \vec{e}_i - \vec{e}_i||_2^2] = \mathbb{E} [|| \sum_{j = \Delta + 1}^{\ell} u_j u_j^T \vec{e}_i||_2^2]
    &= \sum_{j = \Delta + 1}^{\ell} \mathbb{E}[(u_j^T \vec{e}_i)^2]
    = \frac{\ell-\Delta}{\ell}. \label{eq:error_estimate_LS}
\end{align}
The last step above follows since each $u_j$ is uniformly distributed over the sphere of dimension $\ell$. Therefore, we have the $\mathbb{E}[u_{j,0}^2] = 1/\ell$ since $||u_j||_2^2 = 1$ and each $u_{j,k}, k = 0, \dots, \ell-1$ is identically distributed. We conclude that if $\calD_i$ appears $\Delta_i \leq \ell-1$ times, then its error contribution is $\ell-\Delta_i$ and the overall error is therefore $\sum_{i=1}^N \max(0,\ell-\Delta_i)$.

{\bf Complexity analysis:} The time complexity of each least-squares problem is $O(\Delta^2 \ell)$  \cite{boyd} and the $i$-th worker solves at most $\ell \delta_i$ of them independently and in parallel. The marginal cost of this calculation as against the calculation of the actual gradients will be very small in most practical settings. 

\subsection{Illustrative Example}
\label{sec:illustrative_eg}

Consider Fig. \ref{fig:partial_straggler_GD} (from \cite{tayyebehM21}) where the dataset consists of chunks $\calD_1, \dots, \calD_5$ and are assigned in a non-uniform fashion to workers $W_1, \dots, W_5$. 
Suppose that the PS wants the exact gradient with $\ell=2$. As two copies of each chunk have been processed, the PS broadcasts the encode-and-transmit signal and the vector $\psi = [5~2~0~2~3]$ to all the workers which indicates, e.g., that $W_1$ has processed all its chunks, $W_3$ is failed etc. The matrix $\bfB$ of indeterminates for this example and the corresponding $\bfB^{(i)}$'s, turn out to be
\setcounter{MaxMatrixCols}{20}
\begin{align}
    \mathbf{B} &= [\bfB^{(1)} \vert \bfB^{(2)} \vert \bfB^{(3)} \vert \bfB^{(4)}\vert \bfB^{(5)}]\\
    &=
    \begin{bmatrix}
    a_1 & a_2 & \vert & a_3 & a_4 & \vert & a_5 & a_6 & \vert & a_7 & a_8 & \vert & a_9 & a_{10}\\
    b_1 & b_2 & \vert & b_3 & b_4 & \vert & 0   & 0   & \vert &  0  & 0   & \vert & 0 & 0\\
    0  & 0   & \vert & 0   &  0  & \vert & 0   & 0   & \vert &  0  & 0   & \vert & 0 & 0\\
    0   & 0   & \vert & c_3   & c_4   & \vert & c_5 & c_6 & \vert & 0 & 0 & \vert & 0 & 0\\
    d_1 & d_2 & \vert & 0   & 0   & \vert & 0   & 0   & \vert & d_7 & d_8 & \vert & d_{9} & d_{10}
    \end{bmatrix}. \label{eq:indeterminates}
\end{align}
In this example, $W_5$ has processed $\calD_1, \calD_4, \calD_5$ so that $\tilde{B}_5 = [\bfB^{(1)}~\bfB^{(4)}~\bfB^{(5)}]$. Note that the PS requires the vectors $[\vec{e}_0^T~\vec{e}_0^T~\vec{e}_0^T~\vec{e}_0^T~\vec{e}_0^T]$ and $[\vec{e}_1^T~\vec{e}_1^T~\vec{e}_1^T~\vec{e}_1^T~\vec{e}_1^T]$ to lie in the row-space of $\bfB$ so that it can recover $\sum_{i=1}^5 g_{\calD_i}[k], k=0,1$ from the encoded gradients. 
Towards this end, e.g., $W_5$ solves \eqref{eq:ls_min_algo} in Algorithm \ref{alg:find-enc-coeff} with the matrix $\tilde{B}_5$, i.e.,
\begin{align}
    \min || [\bfI_2 \vert \bfI_2 \vert \bfI_2] - \bfR [\bfB^{(1)}~\bfB^{(4)}~\bfB^{(5)}]||_2
\end{align}
where the decision variables are $a_1, a_2,a_7,a_8,a_9, a_{10}$, $b_1, b_2$ and $d_1, d_2, d_7, d_8, d_9, d_{10}$.
For instance, corresponding to the first column of $\bfB^{(1)}$, $W_5$'s problem becomes determining the minimum $\ell_2$-norm solution of the under-determined least-squares problem 
\begin{align}
    \bigg{|}\bigg{|}
    \begin{bmatrix}
        r_{01} & r_{02} & r_{05}\\
        r_{11} & r_{12} & r_{15}
    \end{bmatrix}
    \begin{bmatrix}
        a_1\\
        b_1\\
        d_1
    \end{bmatrix} - \begin{bmatrix}
        1\\
        0
    \end{bmatrix}
    \bigg{|}\bigg{|}_2.
\end{align}
We emphasize that the same minimization will be independently performed at workers $W_1$ and $W_2$ as well. 
After each $W_j$ determines its encoding coefficients and transmits $g_j$ for $j = 1, \dots, 5$, the PS can easily recover $\sum_{j=1}^5 g_{\calD_j}[k] = \sum_{i=1}^5 r_{kj} g_j$ for $k=0,1$.

Algorithm \ref{alg:find-enc-coeff} works as is, even in the case when the PS is interested in an approximate gradient. For instance, suppose that the PS sends the encode-and-transmit signal when the vector $\psi = [4~2~0~2~3]$, so that only one copy of $\calD_5$ has been processed in the cluster. The corresponding encoding coefficients can still be computed using Algorithm \ref{alg:find-enc-coeff}. The only difference will be that the relevant indeterminate matrix becomes 
\setcounter{MaxMatrixCols}{20}
\begin{align}
    \mathbf{B}' 
    &=
    \begin{bmatrix}
    a'_1 & a'_2 & \vert & a'_3 & a'_4 & \vert & a'_5 & a'_6 & \vert & a'_7 & a'_8 & \vert & 0 & 0\\
    b'_1 & b'_2 & \vert & b'_3 & b'_4 & \vert & 0   & 0   & \vert &  0  & 0   & \vert & 0 & 0\\
    0  & 0   & \vert & 0   &  0  & \vert & 0   & 0   & \vert &  0  & 0   & \vert & 0 & 0\\
    0   & 0   & \vert & c'_3   & c'_4   & \vert & c'_5 & c'_6 & \vert & 0 & 0 & \vert & 0 & 0\\
    d'_1 & d'_2 & \vert & 0   & 0   & \vert & 0   & 0   & \vert & d'_7 & d'_8 & \vert & d'_{9} & d'_{10}
    \end{bmatrix}. \label{eq:indeterminates_approx}
\end{align}
This means, e.g., when $W_5$ is trying to find $d'_{9}$, then it will be solving an over-determined least-squares problem: $\bigg{|}\bigg{|} \begin{bmatrix} r_{05}\\ r'_{05} \end{bmatrix} d'_{9} - \begin{bmatrix} 1\\ 0 \end{bmatrix} \bigg{|}\bigg{|}_2$. It is evident, that all the workers can still agree on the same solution.

    

\begin{figure}[t!]
    \centering
    \begin{subfigure}[t]{0.5\textwidth}
        \centering
        \includegraphics[scale = 1.0]{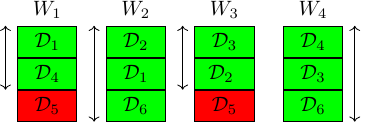}
        \caption{\label{fig:rel_ordering_high}}
    \end{subfigure}%
    ~~~
    \begin{subfigure}[t]{0.5\textwidth}
        \centering
        \includegraphics[scale=1.0]{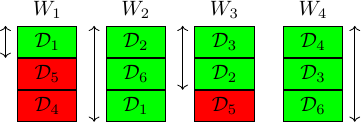}
        \caption{\label{fig:rel_ordering_low}}
    \end{subfigure}
    \caption{\label{fig:rel_orderings} {\small Two different relative orderings of the chunks within workers for the same assignment matrix. Individual figures show the calculation of $Q_5$. Similarly, other $Q_i$ values can be computed. $Q_{\max} = Q_5$ for both assignments. Thus, (a) $Q_{\max} = 10$. (b) $Q_{\max} = 9$.}}
\vs\vs
\end{figure}

\section{Chunk Ordering Algorithm}
\label{sec:chunk_ordering}
Note that the assignment matrix only specifies the assignment of chunks to workers but says nothing about the relative order of chunks within a worker. When taking into account partial stragglers, the relative ordering of the $\calD_i$'s within a worker is crucial. Therefore, an important question when leveraging partial stragglers within GC is one of how to determine this chunk ordering within workers for a given assignment matrix. This will in general depend on models of worker speeds. However,  it is a difficult task to obtain accurate models on the worker speeds within a cloud cluster as conditions can change dynamically. 

We work instead with a combinatorial metric that depends on the total number of chunks that the cluster has to process in the worst case, such that at least a single copy of each $\calD_i$ is processed; this was also used in \cite{c3les,DasR22} in a coded matrix computation context. This metric minimizes the worst case number of chunks that the cluster needs to process in the case when $\ell=1$. In particular, for a given $\calD_i$, let $Q_i$ denote the maximum number of chunks that can be processed by the cluster such that none of the copies of $\calD_i$ are processed (see Figs. \ref{fig:rel_ordering_high} and \ref{fig:rel_ordering_low}) and let $Q_{\max} = \max_{i=1, \dots, N} Q_i$. Thus, $1+Q_{\max}$ is a metric that quantifies the worst-case number of chunks that need to be processed before at least a single copy of every $\calD_i$ is guaranteed to be processed. Here, the worst-case is over the speeds of the different workers. We say that an assignment is optimal if it achieves the lowest possible $Q_{\max}$ for a given assignment matrix. Indeed, Figs. \ref{fig:rel_ordering_high} and \ref{fig:rel_ordering_low} show two different orderings for the same assignment matrix with different values of $Q_{\max}$.

\begin{algorithm}[!t]
\caption{Chunk-Ordering}
\begin{algorithmic}[1]
\REQUIRE Assignment matrix $A$ such that $N=m$ and $\gamma_i = \delta_i = \delta$ for all $i \in [m]$
\ENSURE An optimal ordering matrix $O$.
\STATE Let $\tilde{A}^{(1)} \in \{0,1\}^{m\times m}$ such that $\tilde{A}^{(1)}_{i,j} = \begin{cases}
    1 & \text{~if~} A_{i,j} \neq 0\\
    0 & \text{~otherwise.}
\end{cases}$ and $\delta^{(1)} = \delta$.
\FOR{{$i$ ranges from $1$ to $\delta$}}
        \STATE Apply Claim \ref{claim:bipartite_matching} with $\tilde{A}^{(i)}$ and $\delta^{(i)}$ and solve the max-bipartite matching problem in Claim \ref{claim:bipartite_matching}. Let the permutation matrix be $P_i$.
        \STATE Let $\tilde{A}^{(i+1)}:= \tilde{A}^{(i)}-P_i$ and $\delta^{(i+1)}=\delta^{(i)}-1$.   
\ENDFOR 
\STATE Set $O=\sum_{i\in[\delta]} i P_i$. 
\end{algorithmic}
\label{alg:chunk-ordering}
\end{algorithm}

From the point of view of leveraging partial stragglers, for exact GC it can be shown that the ordering imposed by the cyclic assignment scheme (in \cite{tandon_gradient}) is optimal (({\it cf.} Remark 1 in work \cite{c3les})) for this $Q_{\max}$ metric. However, for approximate GC, when the number of stragglers can be much higher, other assignments, e.g., those based on regular graphs perform better \cite{dimakis_cyclic_mds,charles2018gradient}. In particular, in these cases, the assignment matrix $A$ is chosen as the adjacency matrix of the relevant graphs and is such that $N = m$ and $\gamma_i = \delta_i = \delta$ for $i \in [m]$. For such assignment matrices, Algorithm \ref{alg:chunk-ordering} presents an optimal algorithm for determining the ordering of the chunks within each worker in $Q_{\max}$-metric. Corresponding to the assignment matrix $A$, we can associate an ordering matrix $O$ which is such that $O_{i,j} = 0$ if $A_{i,j} = 0$ and $O_{i,j} \in [\delta]$ otherwise. Thus, if $O_{i,j} = \alpha \neq 0$, it means that $\calD_i$ is the $\alpha$-th chunk in $W_j$, e.g., $\alpha = 1$ implies that the chunk is at the top and $\alpha = \delta$ means that it is at the bottom. 
Let $Q_i(O)$ denote the maximum number of chunks that can be processed across the cluster, such that no copy of $\calD_i$ has been processed. Then, upon inspection, we can see that
\begin{align*}
    Q_i(O) = \sum_{j \in [m]} \mathbf{1}_{ \{O_{i,j} \neq 0\}} (O_{i,j} - 1) + (m-\delta)\delta
    &= \sum_{j \in [m]} \mathbf{1}_{\{O_{i,j} \neq 0\}} O_{i,j} + (m-\delta-1)\delta,
\end{align*}
where the notation $\mathbf{1}_{Y}$ denotes the indicator of $Y$. Next, let $Q_{\max} (O) = \max_{i \in [m]} Q_i(O)$. 
Thus, given an assignment matrix $A$, our goal is to find an ordering matrix, such that $Q_{\max}(O)$ is minimized. As $(m-\delta-1)\delta$ is a constant, this optimization is equivalent to the following min-max problem.
\begin{align}
    \underset{O}{\text{minimize}} \max_{i \in [m]} \sum_{j \in [m]} O_{i,j}.
\end{align}
Each column of $O$ has exactly $\delta$ non-zero entries, with each value in $[\delta]$ appearing exactly once. Counting the sum of the entries in $O$ two different ways yields
\begin{align*}
    m \frac{\delta(\delta+1)}{2} &= \sum_{i \in [m], j \in [m]} O_{i,j} \leq m Q_{\max}(O),
    \text{~so that~}  Q_{\max}(O) \geq \frac{\delta(\delta+1)}{2}.
\end{align*}
It turns out that this bound is in fact achievable. 
Let $\tilde{A}=\tilde{A}^{(1)}$ (defined in Algorithm \ref{alg:chunk-ordering}). Define $\mathbb{G}(\tilde{A}) = (V,E)$  as a bipartite graph on vertex set $V = X \cup Y$ such that $|X|=|Y|=m, X \cap Y = \emptyset$ and $\text{deg}(v) = \delta$ for all $v\in X\cup Y$. For $u \in X, v \in Y$, we have that $(u,v)\in E$ if and only if $\tilde{A}_{u,v} = 1$. Thus, $\tilde{A}$ is in one-to-one correspondence with $\mathbb{G}(\tilde{A})$. Algorithm \ref{alg:chunk-ordering} decomposes $\mathbb{G}(\tilde{A})$ into a collection of disjoint perfect matchings \cite{kleinbergT05} which are then assigned values in $[\delta]$. This gives us the required ordering.

\begin{claim}\label{claim:bipartite_matching}
    Let $\tilde{A}\in \{0,1\}^{m\times m}$ be such that both  $i$-th row sum and $j$-th column sum of $\tilde{A}$ are  $\delta$ for all $i,j\in[m]$. Then there exists a permutation matrix $P\in \{0,1\}^{m\times m}$ such that $\tilde{A}-P \in \{0,1\}^{m\times m}$ and both  $i$-th row sum and $j$-th column sum of $\tilde{A}-P$ are  $\delta-1$ for $i,j\in[m]$.
\end{claim}
\begin{proof}
We claim that $\mathbb{G}(\tilde{A})$ has a $X$-perfect matching. Let $S\subseteq X$ be arbitrary and $N(S) \subseteq Y$ denote its neighborhood. Let $\kappa$ denote the average degree of the vertices in $N(S)$ in the subgraph induced by $S \cup N(S)$.
Then, we have 
\begin{align*}
    \delta |S| &= \kappa |N(S)| \leq \delta |N(S)|, \text{~so that~}  |S| \leq |N(S)|.
\end{align*}
The first inequality above follows because the degree of each node in $N(S)$ in $\mathbb{G}(\tilde{A})$ is $\delta$. Thus, Hall's condition \cite{vanlint_wilson_2001} holds and the claim follows.
Since $|X|=|Y|$, this is actually a perfect matching $M$. 
Then, this perfect matching $M$ gives the desired $P$: $P_{u,v}=1$ if $(u,v)\in M$ and $P_{u,v}=0$ if $(u,v)\notin M$. Removing the matching $M$ from $\mathbb{G}(\tilde{A})$ results in a bi-regular bipartite graph with degree $\delta-1$ which corresponds to $\tilde{A} - P$.
\end{proof}

\begin{remark}
    The proposed algorithm above finds the optimal ordering when considering the case of $\ell=1$ with $N=m$ (number of chunks equal to number of workers); it only applies when $N=m$. While this algorithm is expected to have better performance than a randomly chosen ordering in the case of higher $\ell$, we do not have an optimal construction in this case.
\end{remark}
\textbf{Algorithm \ref{alg:chunk-ordering} Analysis:} 
The algorithm gives $\delta$ permutation matrices $\{P_i\}_{i\in[\delta]}$ such that $\tilde{A} = \sum_{i=1}^\delta P_i$, i.e., the set of non-zero entries of the $P_i$'s are disjoint. Since $O=\sum_{i\in[\delta]} i\cdot P_i$, this implies that each column has exactly $\delta$ non-zero entries such that these entries consists of elements of $[\delta]$.  Therefore, $O$ is an ordering matrix. In addition, each row has $\delta$ non-zero elements from $[\delta]$, so that all row sums and $Q_{\max}(O)$ equal $\frac{\delta(\delta+1)}{2}$. A maximum matching can be in time $O(m^2\delta)$ by converting it to a max-flow problem \cite{kleinbergT05}, so our overall complexity is $O(m^2\delta^2)$. The complexity can potentially be reduced further by using more efficient max-flow algorithms.



\section{Numerical Experiments and Comparisons}
\label{sec:num_exp_comp}
Our proposed protocol in Section \ref{sec:gc_partial_stragglers}, utilizes additional communication between the PS and the workers. We note here that ideas in \cite{tayyebehM21} that are based on Lagrange interpolation can potentially be adapted to arrive at an exact GC protocol within our setting. 
%
%
However, the numerical instability of Lagrange interpolation is a serious issue with this solution, since it can be shown that the degree of the polynomial to be interpolated is at least $m - \ell$. Even for values such as $\ell=2$ and $m = 22, 27, 32$, the error in Lagrange interpolation is too high for the technique to be useful (see Appendix \ref{sec:LIF_instability}).  

\begin{figure}[t!]
    \centering
    \begin{subfigure}[t]{0.5\textwidth}
        \centering
        \includegraphics[scale = 0.25]{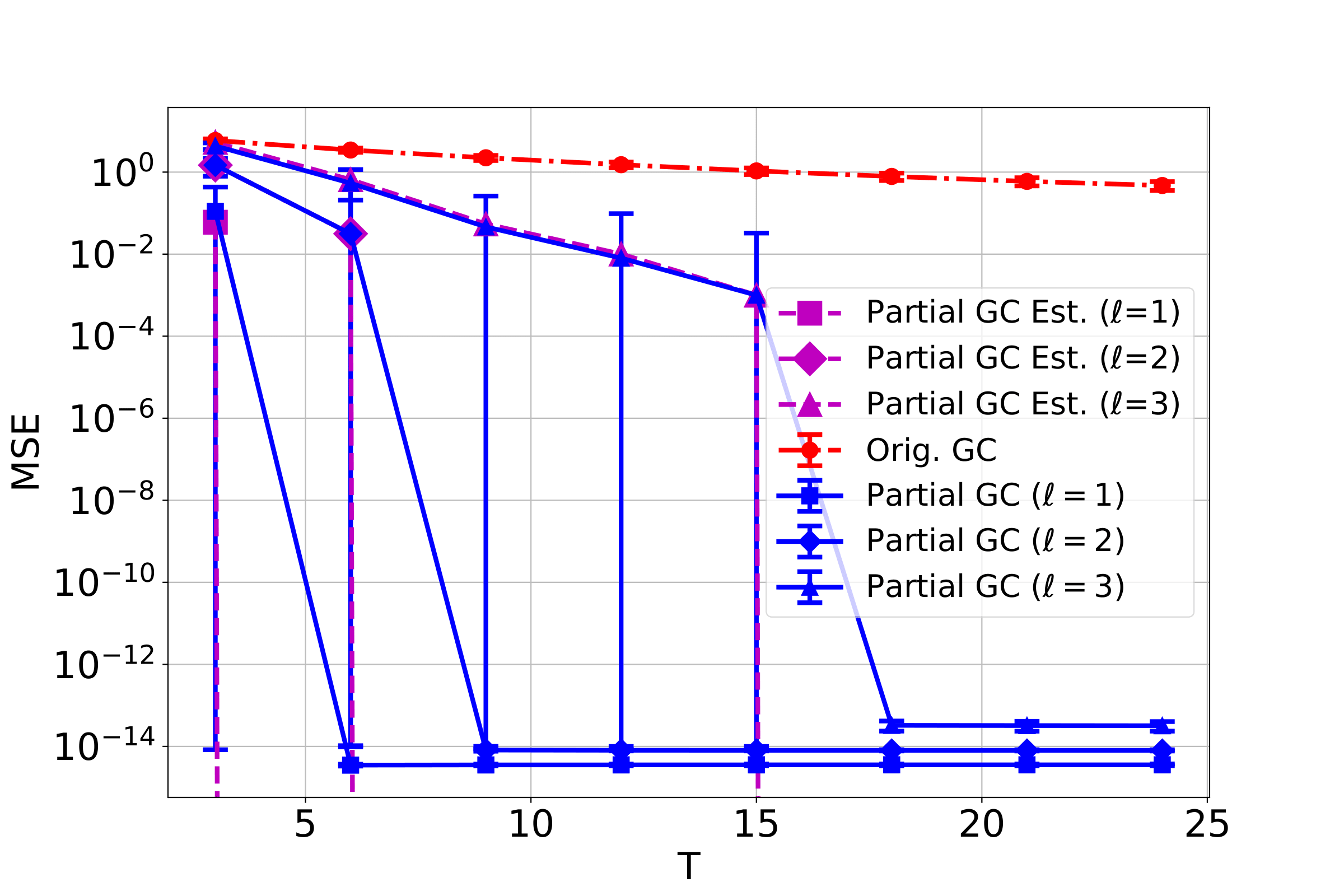}
        \caption{\label{fig:error_G200}}
    \end{subfigure}%
    ~
    \begin{subfigure}[t]{0.5\textwidth}
        \centering
        \includegraphics[scale=0.4]{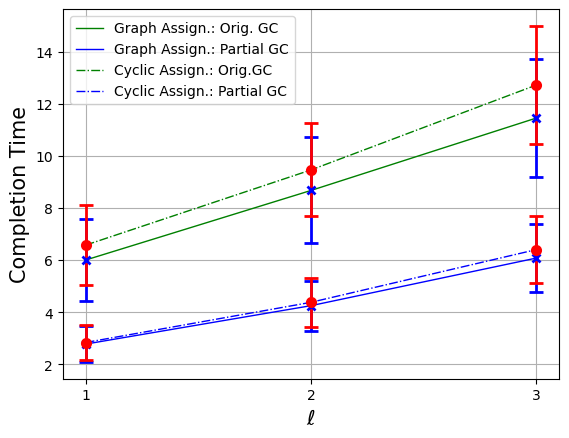}
        \caption{\label{fig:completion_time}}
    \end{subfigure}
    \caption{\label{fig:num_expts} {\small (a) Mean-squared error (MSE) vs. $T$ for an approximate GC scenario. Blue curves: proposed protocol with $\ell=1,2,3$, purple curves: corresponding MSE estimates, and red curve: original GC protocol with $\ell=1$. Error bars correspond to one standard deviation. (b) Completion time vs. $\ell$ for exact GC scenario with two different assignment matrices. Blue curves: proposed protocol, green curves: original GC protocol. Error bars correspond to one standard deviation.}}
\vs\vs
\end{figure}

In what follows, we present comparisons with the original GC protocol via a simulation of the distributed system for both the exact and approximate GC scenarios. All software code for recreating these results can be found at \cite{partialGC_repository}. These simulations were performed within a MacBook Pro (M1 chip, 16 GB RAM). In both scenarios, we simulated node failures and slow-downs. We generated a random vector of length-$m$ where $\alpha$ workers uniformly at random are chosen to be failed ($\alpha$ is chosen based on the scenario). For the other workers, the amount of time taken to process a chunk is chosen i.i.d. from an exponential distribution with parameter $\mu=1$. The entries of the matrix $\bfR$ are chosen i.i.d. from the standard normal $N(0,1)$.

{\noindent {\bf Approximate GC simulation:}} We considered two different random regular graphs, $G_i, i =1,2$ with sizes $m = 200, 300$ and degree $\nu=8$. These graphs have second-largest absolute value of eigenvalues: 5.14 and 5.23, i.e., less than $2\sqrt{\nu-1}$ so they can be considered as Ramanujan graphs \cite{chung_spectral} (these were used in \cite{dimakis_cyclic_mds}). Their adjacency matrices were used as the assignment matrix. The number of failures $\alpha$ was set to $\nu -1$.

For the original GC protocol with arbitrary assignment matrices, \cite{tayyebehM21} provides a communication-efficient approximate GC method that relies on rational interpolation. However, even for $m=100$, these will result in very complex basis functions. Furthermore, there are no numerical results in \cite{tayyebehM21} (or posted software) and the approximation guarantees depend on the form of the function to be interpolated, i.e., the guarantees cannot be expressed in terms of the system parameters. Thus, in our comparison, we only consider the original GC protocol with $\ell =1$.

For the original GC protocol, we compute the least squares solution $\hat{r}$ for minimizing $||Ar - \mathds{1}||^2_2$. Here, $\hat{r}$ such that $\hat{r}_i = 0$ if $W_i$ has not completed processing all its chunks at time $T$. For our proposed protocol, we leverage partial work completed by $T$, as described in Section \ref{sec:gc_partial_stragglers}. The overall error is computed as the sum of the errors for the least-squares solution for each $z_i^T, i = 0, \dots, \ell-1$. Each data point on the curves was chosen by performing 1000 simulations of failures and slow-downs. We have also plotted the expected value of the error, derived in \eqref{eq:error_estimate_LS}.

Fig. \ref{fig:error_G200} shows the mean-squared-error (MSE) comparing the original GC approach with $\ell=1$ and our partial GC approach for $\ell = 1,2,3$ for graph $G_1$ (see Appendix \ref{sec:addl_num_exp} for $G_2$ results). As can be observed, the MSE for our approach is several orders of magnitude lower with increasing $T$. Crucially, our estimate \eqref{eq:error_estimate_LS} closely tracks the behavior of the error of our method. Thus, it can easily be used by the PS as a way to decide when to send the encode-and-transmit message. We emphasize that our approach, even with $\ell \geq 2$ actually has a lower MSE than the original approach (that operates with $\ell=1$). Note that with $\ell \geq 2$, we will enjoy a lower communication time in a real-world distributed cluster. However, as we are working with a simulation of the distributed system, at this time we do not have precise figures for the reduction in communication time on actual clusters.

In Fig. \ref{fig:num_random_expts} in Appendix \ref{sec:addl_num_exp}, we compare the performance of our approach under the optimal ordering ({\it cf.} Section \ref{sec:chunk_ordering}) and an appropriately chosen random ordering. The random ordering is picked as follows. We generated 100 independent random orderings and selected the one with the best (smallest) $Q_{\max}(O)$. Our optimal ordering, which can be found efficiently, clearly has a better performance.


{\noindent {\bf Exact GC simulation:}} We considered (i) the cyclic assignment \cite{tandon_gradient} with $N=m=200$ and $\delta = 8$, and (ii) the assignment based on graph $G_1$ discussed above. The number of failures $\alpha$ in the simulations is set to $\delta -\ell$ so that exact gradient reconstruction is possible. For both approaches, we determined the time $T$ such that each chunk is processed at least $\ell$ times across the cluster (for original GC we only consider workers that have processed all their chunks). These values were averaged over 1000 runs for each value of $\ell$. In Fig. \ref{fig:completion_time} we clearly observe that the exact gradient can be computed using our method is approximately half the time as compared to the original GC protocol. We note that the average completion time for the original GC protocol ($G_1$-based assignment) for $\ell=1$ is about $T=6$ time units. However, the MSE for the original GC protocol in the approximate scenario continues to be high at $T=24$. The reason is that there are $\nu-1$ failures that are introduced in the simulation. The approximate GC recovery operates by solving a least-squares problem for the fixed $G_1$-based assignment. Thus, the MSE does not necessarily drop to zero even if one copy of each chunk has been processed in the cluster. However, there are exact recovery algorithms that one can use in this case. We note here that when $\ell \geq 2$, the known techniques for exact recovery for the original GC protocol are based on Lagrange interpolation and are numerically unstable. Thus, the gradient recovered using the original approach will in general not be useful.

Fig. \ref{fig:rand_comp_time} in Appendix \ref{sec:addl_num_exp} compares the completion times of random vs. optimal ordering; the optimal ordering is clearly better. 


\section{Limitations of our work}
 Our chunk ordering algorithm (Section \ref{sec:chunk_ordering}) is optimal only in the case when the assignment matrix has $N=m$ and $\gamma_i = \delta_i = \delta$ for $i \in [m]$. While several well known gradient coding schemes, especially those from regular graphs, satisfy this property, there are others that do not. Our results in Section \ref{sec:num_exp_comp} with $\ell \geq 2$ are simulations of the actual distributed cluster and indicate lower MSE than the original GC protocol. However, we do not have actual cloud platform statistics on the reduction in communication time within our method. We do however expect this reduction to be quite significant.
\section{Conclusions}

We presented a novel gradient coding protocol that leverages the work performed by partial stragglers. Our protocol is simultaneously computation-efficient and communication-efficient, and numerically stable. Furthermore, we present rigorous analyses of the protocol's correctness and performance guarantees. Our protocol is one of the first to provide a satisfactory solution to the problem of communication-efficient, approximate gradient coding for general assignment matrices. 


\bibliographystyle{IEEETran}


\begin{thebibliography}{10}
\providecommand{\url}[1]{#1}
\csname url@samestyle\endcsname
\providecommand{\newblock}{\relax}
\providecommand{\bibinfo}[2]{#2}
\providecommand{\BIBentrySTDinterwordspacing}{\spaceskip=0pt\relax}
\providecommand{\BIBentryALTinterwordstretchfactor}{4}
\providecommand{\BIBentryALTinterwordspacing}{\spaceskip=\fontdimen2\font plus
\BIBentryALTinterwordstretchfactor\fontdimen3\font minus
  \fontdimen4\font\relax}
\providecommand{\BIBforeignlanguage}[2]{{%
\expandafter\ifx\csname l@#1\endcsname\relax
\typeout{** WARNING: IEEEtran.bst: No hyphenation pattern has been}%
\typeout{** loaded for the language `#1'. Using the pattern for}%
\typeout{** the default language instead.}%
\else
\language=\csname l@#1\endcsname
\fi
#2}}
\providecommand{\BIBdecl}{\relax}
\BIBdecl

\bibitem{GoodBengCour16}
I.~J. Goodfellow, Y.~Bengio, and A.~Courville, \emph{Deep Learning}.\hskip 1em
  plus 0.5em minus 0.4em\relax Cambridge, MA, USA: MIT Press, 2016,
  \url{http://www.deeplearningbook.org}.

\bibitem{jain2013lrmc}
P.~Jain, P.~Netrapalli, and S.~Sanghavi, ``Low-rank matrix completion using
  alternating minimization,'' in \emph{Proceedings of the forty-fifth Annual
  ACM Symposium on Theory of Computing (STOC)}, 2013, pp. 665--674.

\bibitem{li2014scaling}
M.~Li, D.~G. Andersen, J.~W. Park, A.~J. Smola, A.~Ahmed, V.~Josifovski,
  J.~Long, E.~J. Shekita, and B.-Y. Su, ``Scaling distributed machine learning
  with the parameter server,'' in \emph{11th USENIX Symposium on operating
  systems design and implementation (OSDI 14)}, 2014, pp. 583--598.

\bibitem{GoogleCP}
\BIBentryALTinterwordspacing
``{Google Spot Virtual Machines}.'' [Online]. Available:
  \url{https://cloud.google.com/spot-vms}
\BIBentrySTDinterwordspacing

\bibitem{tandon_gradient}
R.~Tandon, Q.~Lei, A.~G. Dimakis, and N.~Karampatziakis, ``Gradient coding:
  Avoiding stragglers in distributed learning,'' in \emph{Intl. Conf. Mach.
  Learn. (ICML)}, August 2017, pp. 3368--3376.

\bibitem{dimakis_cyclic_mds}
N.~Raviv, R.~Tandon, A.~Dimakis, and I.~Tamo, ``Gradient coding from cyclic
  {MDS} codes and expander graphs,'' in \emph{Intl. Conf. Mach. Learn. (ICML)},
  July 2018, pp. 4302--4310.

\bibitem{reedsolomon_GC18}
W.~Halbawi, N.~Azizan, F.~Salehi, and B.~Hassibi, ``Improving distributed
  gradient descent using reed-solomon codes,'' in \emph{IEEE Intl. Symp. on
  Info. Th.}, 2018, pp. 2027--2031.

\bibitem{chen2018lag}
T.~Chen, G.~Giannakis, T.~Sun, and W.~Yin, ``Lag: Lazily aggregated gradient
  for communication-efficient distributed learning,'' \emph{Adv. in Neural Inf.
  Process. Syst. (NeurIPS)}, vol.~31, 2018.

\bibitem{tieredGC20}
S.~Sasi, V.~Lalitha, V.~Aggarwal, and B.~S. Rajan, ``Straggler mitigation with
  tiered gradient codes,'' \emph{IEEE Trans. on Comm.}, vol.~68, no.~8, pp.
  4632--4647, 2020.

\bibitem{dynamicClusteringGC23}
B.~Buyukates, E.~Ozfatura, S.~Ulukus, and D.~Gündüz, ``Gradient coding with
  dynamic clustering for straggler-tolerant distributed learning,'' \emph{IEEE
  Trans. on Comm.}, vol.~71, no.~6, pp. 3317--3332, 2023.

\bibitem{numerically_stableGC20}
N.~Charalambides, H.~Mahdavifar, and A.~O. Hero, ``Numerically stable binary
  gradient coding,'' in \emph{IEEE Intl. Symp. on Info. Th.}, 2020, pp.
  2622--2627.

\bibitem{bottou_optimization}
L.~Bottou, F.~E. Curtis, and J.~Nocedal, ``Optimization methods for large-scale
  machine learning,'' \emph{SIAM Review}, vol.~60, no.~2, pp. 223--311, May
  2018.

\bibitem{charles2017approximate}
Z.~Charles, D.~Papailiopoulos, and J.~Ellenberg, ``Approximate gradient coding
  via sparse random graphs,'' 2017 [Online] arxiv:1711.06771.

\bibitem{glasgowW21}
M.~Glasgow and M.~Wootters, ``Approximate gradient coding with optimal
  decoding,'' \emph{IEEE J. Select. Areas Info. Th.}, vol.~2, no.~3, pp.
  855--866, 2021.

\bibitem{kadheKR19}
S.~Kadhe, O.~O. Koyluoglu, and K.~Ramchandran, ``Gradient coding based on block
  designs for mitigating adversarial stragglers,'' in \emph{IEEE Intl. Symp. on
  Info. Th.}, 2019, pp. 2813--2817.

\bibitem{soft_BIBD_GC22}
A.~Sakorikar and L.~Wang, ``Soft bibd and product gradient codes,'' \emph{IEEE
  J. Select. Areas Info. Th.}, vol.~3, no.~2, pp. 229--240, 2022.

\bibitem{alistarh2017qsgd}
D.~Alistarh, D.~Grubic, J.~Li, R.~Tomioka, and M.~Vojnovic, ``Qsgd:
  Communication-efficient sgd via gradient quantization and encoding,''
  \emph{Adv. in Neural Inf. Process. Syst. (NeurIPS)}, vol.~30, 2017.

\bibitem{wang2023cocktailsgd}
J.~Wang, Y.~Lu, B.~Yuan, B.~Chen, P.~Liang, C.~De~Sa, C.~Re, and C.~Zhang,
  ``Cocktailsgd: Fine-tuning foundation models over 500mbps networks,'' in
  \emph{Intl. Conf. Mach. Learn. (ICML)}, 2023, pp. 36\,058--36\,076.

\bibitem{Belkin_2019}
M.~Belkin, D.~Hsu, S.~Ma, and S.~Mandal, ``Reconciling modern machine-learning
  practice and the classical bias–variance trade-off,'' \emph{Proceedings of
  the National Academy of Sciences}, vol. 116, no.~32, p. 15849–15854, Jul.
  2019.

\bibitem{YeA18}
M.~Ye and E.~Abbe, ``Communication-computation efficient gradient coding,'' in
  \emph{Intl. Conf. Mach. Learn. (ICML)}, 2018, pp. 5610--5619.

\bibitem{tayyebehM21}
T.~Jahani-Nezhad and M.~A. Maddah-Ali, ``Optimal communication-computation
  trade-off in heterogeneous gradient coding,'' \emph{IEEE J. Select. Areas
  Info. Th.}, vol.~2, no.~3, pp. 1002--1011, 2021.

\bibitem{Pan16}
V.~Pan, ``{How Bad Are Vandermonde Matrices?}'' \emph{SIAM Jour. on Mat.
  Analysis and Appl.}, vol.~37, no.~2, pp. 676--694, 2016.

\bibitem{kadheKR20}
S.~Kadhe, O.~O. Koyluoglu, and K.~Ramchandran, ``Communication-efficient
  gradient coding for straggler mitigation in distributed learning,'' in
  \emph{IEEE Intl. Symp. on Info. Th.}, 2020, pp. 2634--2639.

\bibitem{TauzD19}
L.~Tauz and L.~Dolecek, ``Multi-message gradient coding for utilizing
  non-persistent stragglers,'' in \emph{2019 53rd Asilomar Conference on
  Signals, Systems, and Computers}, 2019, pp. 2154--2159.

\bibitem{adaptiveGC22}
H.~Cao, Q.~Yan, X.~Tang, and G.~Han, ``Adaptive gradient coding,''
  \emph{IEEE/ACM Trans. on Networking}, vol.~30, no.~2, pp. 717--734, 2022.

\bibitem{hetero_aware_GC22}
H.~Wang, S.~Guo, B.~Tang, R.~Li, Y.~Yang, Z.~Qu, and Y.~Wang,
  ``Heterogeneity-aware gradient coding for tolerating and leveraging
  stragglers,'' \emph{IEEE Trans. on Comp.}, vol.~71, no.~4, pp. 779--794,
  2022.

\bibitem{opt_block_coord_GC21}
Q.~Wang, Y.~Cui, C.~Li, J.~Zou, and H.~Xiong, ``Optimization-based block
  coordinate gradient coding,'' in \emph{2021 IEEE Global Communications
  Conference (GLOBECOM)}, 2021.

\bibitem{anytime_GC17}
N.~Ferdinand, B.~Gharachorloo, and S.~C. Draper, ``Anytime exploitation of
  stragglers in synchronous stochastic gradient descent,'' in \emph{IEEE Intl.
  Conf. on Mach. Learning and Appl. (ICMLA)}, 2017, pp. 141--146.

\bibitem{lincostello}
S.~Lin and D.~J. Costello, \emph{Error Control Coding, 2nd Ed.}\hskip 1em plus
  0.5em minus 0.4em\relax Upper Saddle River: Prentice Hall, 2004.

\bibitem{CoverThomas2006}
T.~M. Cover and J.~A. Thomas, \emph{Elements of Information Theory 2nd Edition
  (Wiley Series in Telecommunications and Signal Processing)}.\hskip 1em plus
  0.5em minus 0.4em\relax Wiley-Interscience, 2006.

\bibitem{horn_matrix_analysis}
R.~A. Horn and C.~R. Johnson, \emph{Matrix Analysis}.\hskip 1em plus 0.5em
  minus 0.4em\relax Cambridge University Press, 2012.

\bibitem{ChenD05}
Z.~Chen and J.~J. Dongarra, ``Condition numbers of gaussian random matrices,''
  \emph{SIAM Journal on Matrix Analysis and Applications}, vol.~27, no.~3, pp.
  603--620, 2005.

\bibitem{versh_book}
R.~Vershynin, \emph{High-dimensional probability: An introduction with
  applications in data science}.\hskip 1em plus 0.5em minus 0.4em\relax
  Cambridge University Press, 2018, vol.~47.

\bibitem{boyd}
S.~Boyd and L.~Vandenberghe, \emph{Convex Optimization}.\hskip 1em plus 0.5em
  minus 0.4em\relax Cambridge University Press, 2004.

\bibitem{c3les}
A.~B. Das, L.~Tang, and A.~Ramamoorthy, ``${C}^3{LES}$ : Codes for coded
  computation that leverage stragglers,'' in \emph{IEEE Info. Th. Workshop},
  2018, pp. 1--5.

\bibitem{DasR22}
A.~B. Das and A.~Ramamoorthy, ``Coded sparse matrix computation schemes that
  leverage partial stragglers,'' \emph{IEEE Trans. on Info. Th.}, vol.~68,
  no.~6, pp. 4156--4181, 2022.

\bibitem{charles2018gradient}
Z.~Charles and D.~Papailiopoulos, ``Gradient coding via the stochastic block
  model,'' 2018 [Online] arxiv:1805.10378 [stat.ML].

\bibitem{kleinbergT05}
J.~Kleinberg and E.~Tardos, \emph{Algorithm Design}.\hskip 1em plus 0.5em minus
  0.4em\relax Pearson, 2005.

\bibitem{vanlint_wilson_2001}
J.~H. Van~Lint and R.~M. Wilson, \emph{A Course in Combinatorics}.\hskip 1em
  plus 0.5em minus 0.4em\relax New York: Cambridge University Press, 2001.

\bibitem{partialGC_repository}
A.~Ramamoorthy, R.~Meng, and V.~S. Girimaji, ``{Leveraging partial stragglers
  within gradient coding - software repository},''
  \url{https://github.com/flamethrower775/Leveraging-partial-stragglers-within-gradient-coding},
  2024.

\bibitem{chung_spectral}
F.~R.~K. Chung, \emph{Spectral Graph Theory}.\hskip 1em plus 0.5em minus
  0.4em\relax Rhode Island: American Mathematical Society, 1997.

\end{thebibliography}
\newpage
\appendices
\section{Numerical Instability of Lagrange Interpolation.}
\label{sec:LIF_instability}
Even for values such as $\ell=2$ and $m = 22, 27, 32$, the error in Lagrange interpolation is too high for the technique of \cite{tayyebehM21} to be useful. To see this consider Fig. \ref{fig:lif_error} which shows the error (averaged over 100 random trials) in first interpolating and then evaluating the interpolated Lagrange polynomial (degrees 20, 25 and 30) at specific points (as is done in \cite{tayyebehM21}); the $x$-axis is the precision. It can be observed that the error even with full-precision is too high for the technique to be useful.

\begin{figure}[!b]
\vs\vs
\begin{center} 
\includegraphics[scale=0.5]{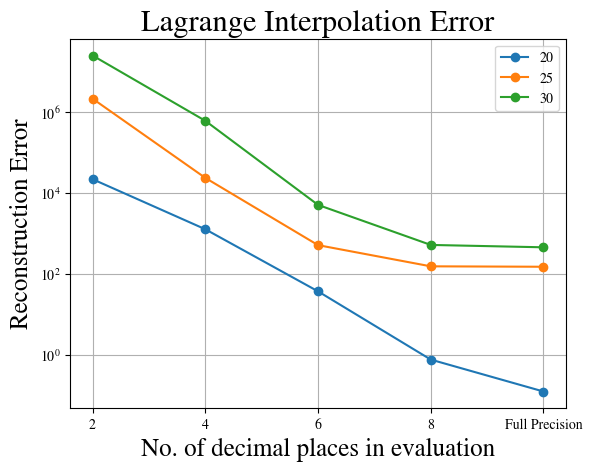}
\caption{\label{fig:lif_error} {\small Error in Lagrange interpolation vs. the number of decimal places (precision) in the evaluation values. The three curves correspond to polynomials of degree 20, 25 and 30 (average of 100 trials). }}
\end{center}
\end{figure}

\section{Additional numerical experiments}
\label{sec:addl_num_exp}

Fig. \ref{fig:error_G300} shows the mean-squared-error (MSE) comparing the original GC approach with $\ell=1$ and our partial GC approach for $\ell = 1,2,3$ for graph $G_2$ with $m=300$ vertices. The results are similar in spirit to the results for the case of $G_1$ that has 200 vertices. Namely, our MSE is orders of magnitude lower than the original GC approach, even when we consider $\ell \geq 2$.

\begin{figure}[!b]
\begin{center} 
\includegraphics[scale=0.25]{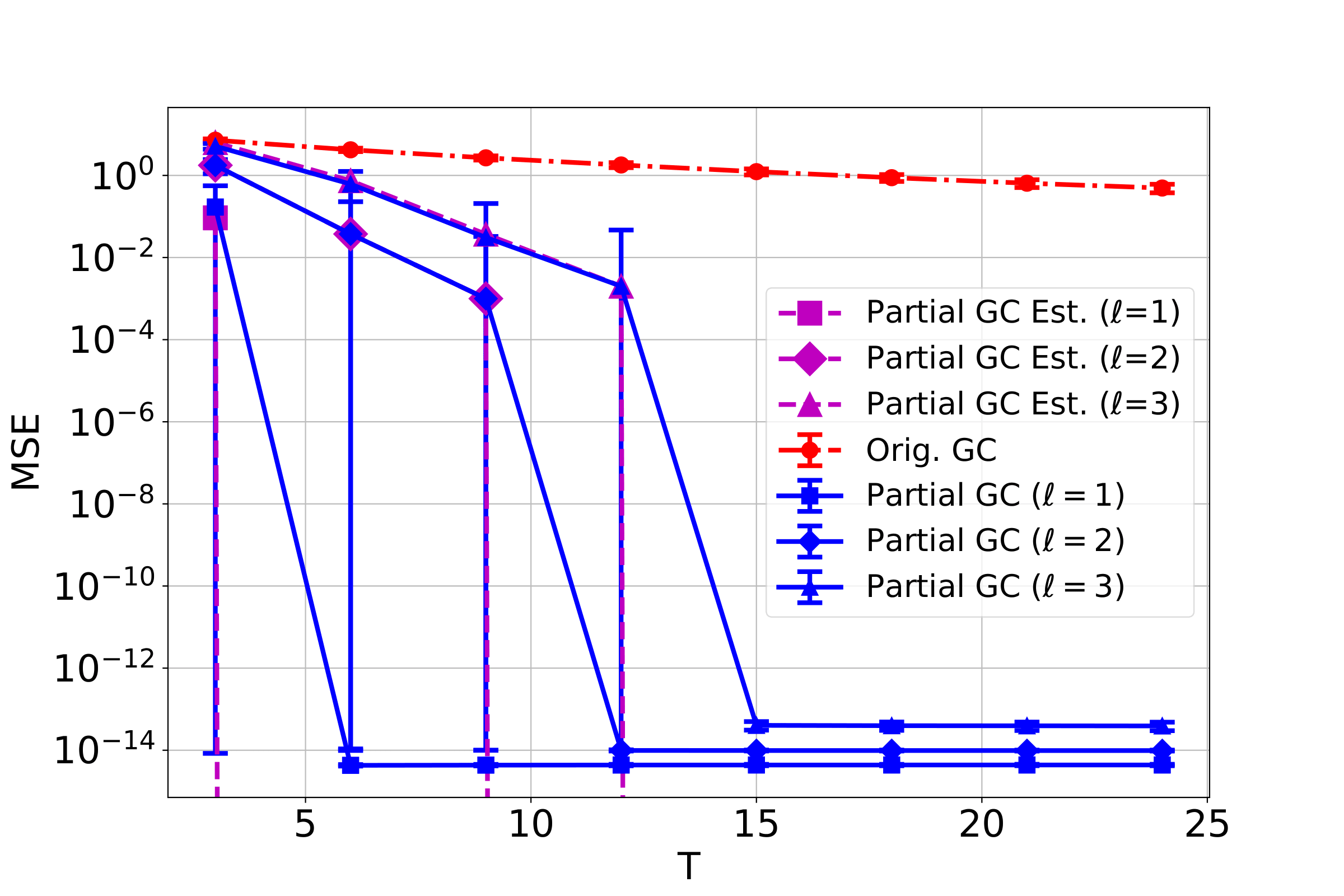}
\caption{\label{fig:error_G300} {\small Mean-squared error (MSE) vs. $T$ for an approximate GC scenario corresponding to an assignment matrix stemming from graph $G_2$. Error bars correspond to one standard deviation. Blue curves: proposed protocol with $\ell=1,2,3$, purple curves: corresponding MSE estimates, and red curve: original GC protocol with $\ell=1$. }}
\end{center}
\end{figure}

In Figs. \ref{fig:error_random_G200} and \ref{fig:error_random_G300} we study the impact of chunk ordering within the workers for the assignment matrices corresponding to graphs $G_1$ and $G_2$ respectively. Each data point on the curves corresponds to 1000 simulations (setup described in Section \ref{sec:num_exp_comp}). In particular, in Fig. \ref{fig:error_random_G200} corresponding to the case of $\ell=1,2,3$ for $G_1$, we observe that the MSE of the optimal ordering consistently remains lower than the MSE of the random ordering and can in fact be multiple orders of magnitude lower when the encode-and-transmit signal is sent at certain time intervals. A similar pattern can be observed in Fig. \ref{fig:error_random_G300}, which shows the case of $\ell=1,2,3$ and $G_2$. 

\begin{figure}[t!]
    \centering
    \begin{subfigure}[t]{0.5\textwidth}
        \centering
        \includegraphics[scale = 0.25]{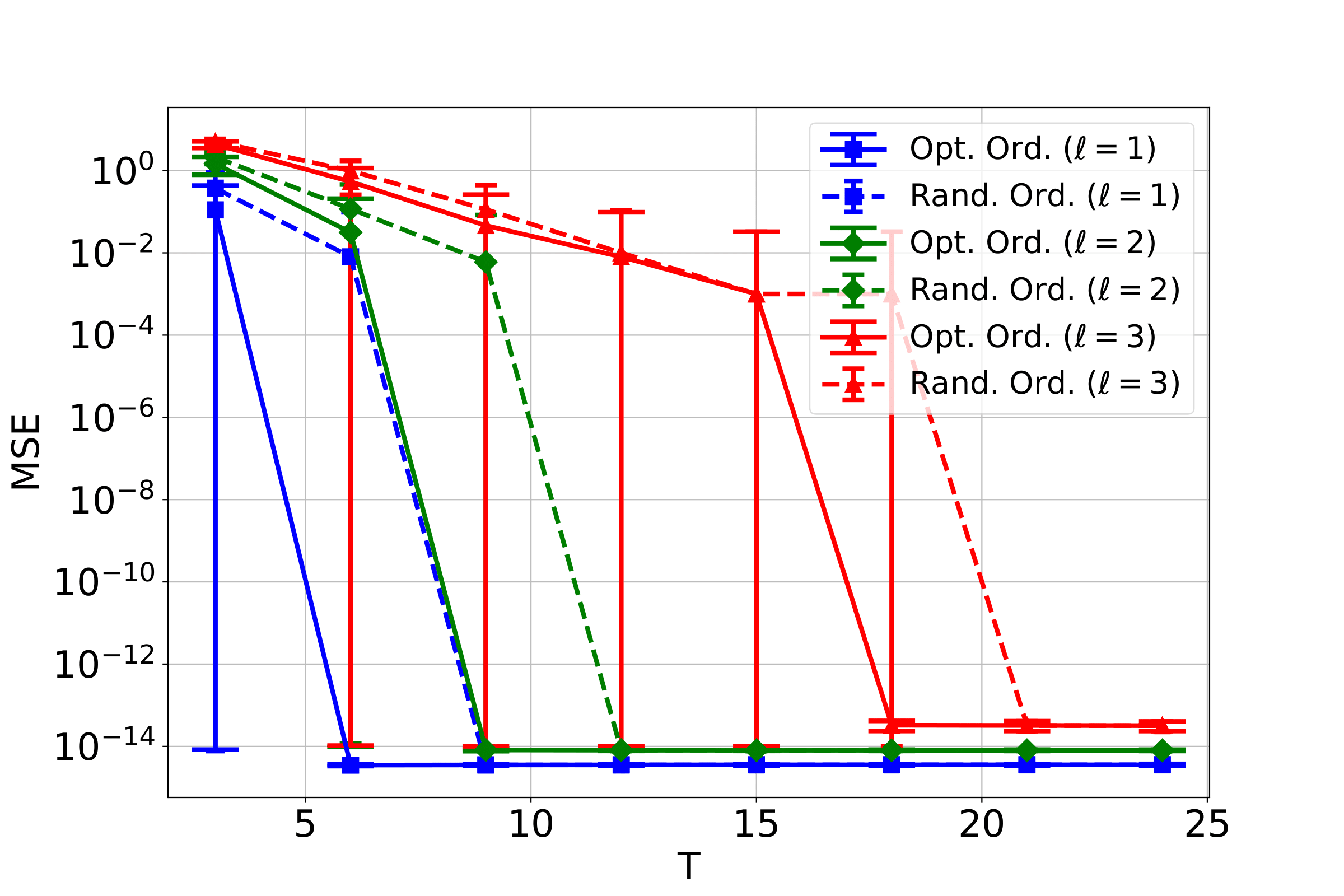}
        \caption{\label{fig:error_random_G200}}
    \end{subfigure}%
    ~
    \begin{subfigure}[t]{0.5\textwidth}
        \centering
        \includegraphics[scale=0.25]{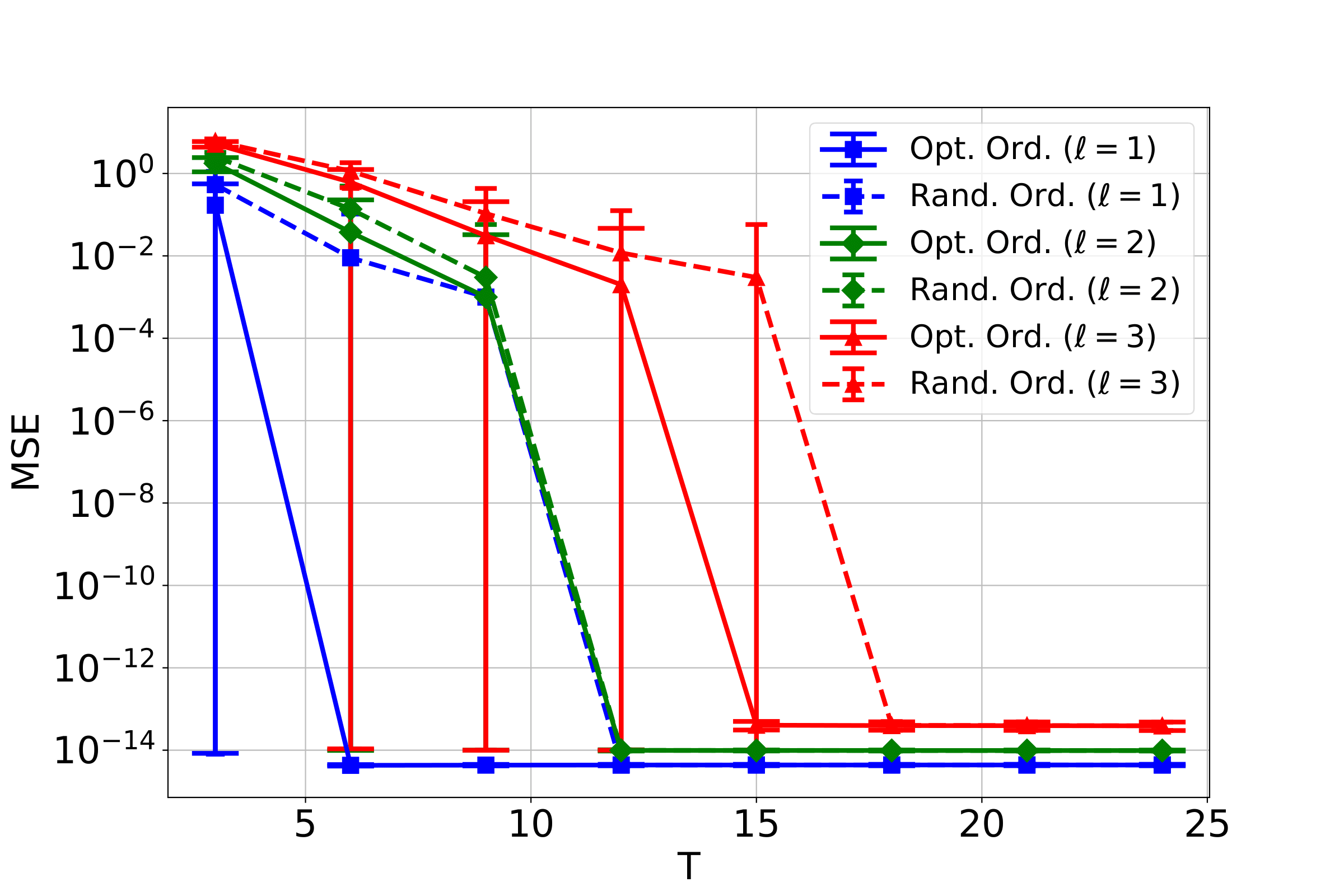}
        \caption{\label{fig:error_random_G300}}
    \end{subfigure}
    \caption{\label{fig:num_random_expts} {\small  Mean-squared error (MSE) vs. $T$ for the approximate GC scenario when considering random chunk ordering and optimal chunk ordering within our protocol with $\ell=1,2,3$. Error bars correspond to one standard deviation. (a) Assignment matrix corresponding to graph $G_1$. (a) Assignment matrix corresponding to graph $G_2$.}}
\end{figure}
Fig. \ref{fig:rand_comp_time} shows the results of a similar experiment comparing the random chunk ordering and our optimal chunk ordering in terms of completion time for exact GC. The assignment matrix corresponds to the graph $G_1$ discussed in Section \ref{sec:num_exp_comp}. 
The random ordering is chosen by sampling 100 independent random orderings and selecting one with the smallest $Q_{max}(O)$. A data point is then generated by running 1000 simulations with the selected random ordering. The results indicate that in an average sense, the completion time of the optimal ordering is clearly lower.

\begin{figure}[!h]
\begin{center} 
\includegraphics[scale=0.5]{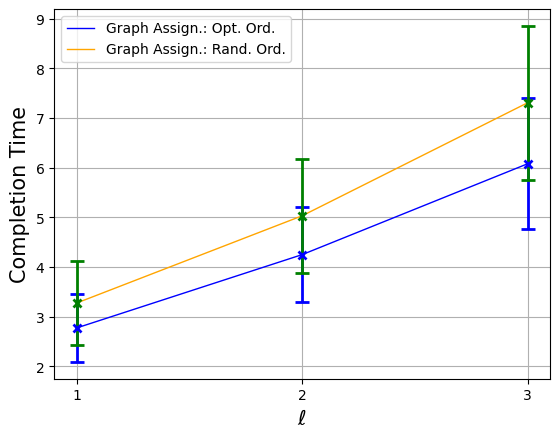}
\caption{\label{fig:rand_comp_time} {\small Completion time vs. $\ell$ for the exact GC scenario with random chunk ordering and optimal chunk ordering. Error bars correspond to one standard deviation.}}
\end{center}
\end{figure}

\clearpage    

\end{document}